\documentclass[preprint,11pt]{article}

\usepackage{amssymb,amsfonts,amsmath,amsthm,amscd,dsfont,mathrsfs}
\usepackage{graphicx,float,psfrag,epsfig,color}

\newcommand\independent{\protect\mathpalette{\protect\independent}{\perp}} 
\def\independent#1#2{\mathrel{\rlap{$#1#2$}\mkern2mu{#1#2}}} 

\newcommand{\supp}{\mathrm{supp}} 
\newcommand{\M}{E_m} 
\newcommand{\F}{\mathbb{F}} 
\newcommand{\mR}{\mathbb{R}} 

\newcommand{\mZ}{\mathbb{Z}}

\newcommand{\pp}{\mathbb{P}}
\newcommand{\E}{\mathbb{E}}

\newcommand{\e}{\varepsilon}

\newcommand{\X}{\mathcal{X}}
\newcommand{\Y}{\mathcal{Y}}

\newcommand{\I}{\mathcal{I}}
\newcommand{\D}{\mathcal{D}}
\newcommand{\C}{\mathcal{C}}
\newcommand{\B}{\mathcal{B}}
\newcommand{\Po}{\mathcal{P}}

\theoremstyle{definition}
\newtheorem{definition}{Definition}
\theoremstyle{plain}
\newtheorem{thm}{Theorem}
\theoremstyle{plain}
\theoremstyle{plain}
\newtheorem{lemma}{Lemma}
\theoremstyle{plain}
\newtheorem{corol}{Corollary}
\theoremstyle{plain}
\theoremstyle{remark}
\newtheorem{remark}{Remark}
\theoremstyle{plain}

\begin{document}
\title{Mutual information, matroids and extremal dependencies  
}
\author{Emmanuel Abbe}
\date{}
\maketitle

\begin{abstract}
In this paper, it is shown that the rank function of a matroid can be represented by a ``mutual information function'' if and only if the matroid is binary. 
The mutual information function considered is the one measuring the amount of information between the inputs (binary uniform) and the output of a multiple access channel (MAC). 
Moreover, it is shown that a MAC whose mutual information function is integer valued is ``equivalent'' to a linear deterministic MAC, in the sense that it essentially contains at the output no more information than some linear forms of the inputs. These notes put emphasis on the connection between mutual information functionals and rank functions in matroid theory, without assuming prior knowledge on these two subjects. The first section introduces mutual information functionals, the second section introduces basic notions of matroid theory, and the third section connects these two subjects.    
It is also shown that entropic matroids studied in the literature correspond to specific cases of MAC matroids.  
\end{abstract}

\section{Information Measures}

\begin{definition} Let $\X$ and $\Y$ be two finite sets called respectively the input and output alphabets and let $M(\X)$ denote the set of probability measures on $\X$.
A channel $W$ with input alphabet $\X$ and output alphabet $\Y$ is a collection of conditional probability measures $\{W(\cdot |x) \in M(\Y): x \in \X\}$.  For fixed alphabets, we denote the set of channels by $M(\Y|\X)$.
\end{definition}

\begin{definition}
The mutual information of a probability measure $\mu \in M(\X \times \Y)$ is defined by
\begin{align*}
& I(\mu) = D(\mu || \mu_\X \times \mu_\Y )= \E_{\mu} \log\frac{\mu}{\mu_\X \times \mu_\Y} ,
\end{align*}
where $\mu_\X$ and $\mu_\Y$ are respectively the marginals in $\X$ and $\Y$ of $\mu$.

If $X$ and $Y$ are two random variables on respectively $\X$ and $\Y$, then $I(X;Y)$ denotes $I(\mu)$ where $\mu$ is the joint distribution of $X,Y$.

If $P \in M(\X)$ is an input distribution and $W \in M(\Y|\X)$ is a channel, then $I(P, W)$ denotes $I(\mu)$ where $\mu= P \circ W$. 

The uniform mutual information (UMI) of a channel $W \in M(\Y|\X)$ is given by $I(W):=I(U_\X \circ W)$, where $U_\X$ is the uniform distribution on $\X$.
\end{definition}





For a given channel $W$ and for any input distribution $P_X$, $I(P_X,W)$ has the following operational meaning in information theory: it is an achievable rate for reliable communication on a discrete memoryless channel with transition probability $W$. In particular, $I(W)$ is an achievable rate and the largest achievable rate is given by the capacity $C=\max_{P \in M(\X)} I(P,W)$.

\begin{definition}
A multiple access channel (MAC) $W$ with $m$ users, input alphabet $\X$ and output alphabet $\Y$, is a channel having input alphabet $\X^m$ and output alphabet $\Y$, i.e., an element of $M(\Y|\X^m)$.
A binary MAC is a MAC for which $\X = \F_2$.
\end{definition}

Let $E_m = \{1,\ldots,m\}$.
\begin{definition}\label{mif}
The mutual information function (MIF) of a MAC $W \in M(\Y|\X^m)$ with input distributions $P_1,\ldots,P_m \in M(\X)$ is defined by the function
\begin{align}
 I(P_1,\ldots,P_m, W)  :\, 2^{E_m} & \rightarrow \mR \notag \\
 S \,\,& \mapsto I[S](P_1,\ldots,P_m,W):=I(X[S];Y,X[S^c]), \label{Ifct}
\end{align}
where $$(X[E_m], Y) \sim (P_1 \times \ldots \times P_m) \circ W.$$
If $P_1=\ldots=P_m=U_\X$, we call this function the uniform mutual information function (UMIF) and we denote it by $I(W)$ (the same notation is used for the single-user mutual information, which is not a conflicting notation since single-user channels correspond to 1-user MACs).
\end{definition}

The operational meaning of the MIF is the following: the region $$\{(R_1,\ldots,R_m) : 0 \leq \sum_{i \in S} R_i \leq  I(X[S];YX[S^c]), \, S \subseteq E_m \}$$ represent achievable rates on a memoryless MAC $W$, when the $m$ users are {\it not} allowed to cooperate during the communication. (If the $m$ users were allowed to cooperate, rates given by $I(P,W)$ for any $P\in M(\X^m)$ would be achievable.) If there are no restriction on the input distributions, the closure of the convex hull of all such regions (for any input distributions) gives the capacity region.

\section{Matroids}
\begin{definition}
A matroid $M$ is an ordered pair $(E, \I)$, where $E$ is a finite set called the ground set and $\I$ is a collection of a subsets of $E$ called the independent sets, which satisfies:
\begin{align*}
 (I1) \quad & \emptyset \in \I .\\
 (I2) \quad & \text{If $I  \in \I$ and $I^\prime \subseteq I$, then $I^\prime \in \I$.} \\
 (I3) \quad & \text{If $I_1,I_2  \in \I$ and $|I_1| < |I_2|$, then there exists an element $e \in I_2 - I_1$} \\
 & \text{such that $I_1 \cup \{e\} \in \I$.} 
\end{align*}
We then say that $M$ is a matroid on $E$ with independent sets $\I$.
\end{definition}
\begin{definition}
Let $M$ be a matroid given by $(E,\I)$. 
\begin{itemize}
\item A basis is a maximal (with respect to the inclusion) subset of $E$ which is independent. The collection of bases is denoted by $\B$. Note that all the subsets of the bases are the independent sets. Hence, a matroid can be defined by its bases. 
\item An dependent set is a subset of $E$ which is not independent. The collection of dependent sets is denoted by $\D=\I^c$.
\item A circuit is a minimal (w.r. to the inclusion) subset of $E$ which is dependent. The collection of circuits is denoted by $\C$.
\end{itemize}
\end{definition}

\begin{definition}
On any matroid $M$, we define a rank function $r: \Po(E) \rightarrow \mZ_+$ such that for any $S \subseteq E$, $r(S)$ is given by the cardinality of a maximal independent set contained in (or equal to) $S$.
\end{definition}
Note: one should check that this is a well defined function, i.e., that any two maximal independent sets in $S$ have the same cardinality. This is actually due to the fact that all the bases in a matroid have the same cardinality. This also implies that $r(E)$ is given by the cardinality of a basis. We denote $R:=r(E)$.
\begin{lemma}
The rank function satisfies the following properties.
\begin{align*}
 (R1) \quad &\text{If $X \subseteq E$, then $r(X)\leq |X|$} .\\
 (R2) \quad &\text{If $X_1 \subseteq X_2 \subseteq E$, then $ r(X_1) \leq r(X_2)$}. \\
 (R3) \quad & \text{If $X_1 , X_2 \subseteq E$, then } \\
 & r(X_1 \cup X_2) + r(X_1 \cap X_2 ) \leq r(X_1) + r(X_2).
\end{align*}
\end{lemma}
\noindent
Note: all the objects that we have defined so far (independent sets, dependent sets, bases, circuits, rank function) can be used to define a matroid, i.e., we can define a matroid as a ground set $E$ with a collection of circuits or a ground set $E$ with a rank function, etc. Moreover, each of these objects can be characterized by a set of axioms, as for example in the following lemma.

\begin{lemma}
Let $E$ be a finite set and $r: \Po (E) \rightarrow \mZ_+$. We have that $r$ is a rank function of a matroid on $E$ if and only if $r$ satisfies (R1), (R2) and (R3). 
\end{lemma}

\begin{definition}
A vector matroid over a field $F$ is a matroid whose ground set is given by the column index set of a matrix $A$ defined over $F$, and whose independent sets are given by the column index subsets indicating linearly independent columns. We denote such a matroid by $M=M[A]$. We call $A$ a representative matrix of the matroid. 
\end{definition}

For a vector matroid, the objects defined previously (dependent sets, bases, rank function) naturally match with the objects defined by the corresponding linear algebraic definition. The matroid theory is also connected to other fields such as graph theory. For an undirected graph, the set of edges define a ground set and a collection of edges that does not contain a cycle defines an independent set. 
A major problem in matroid theory, consist in identifying whether a given matroid belongs to a certain class of structured matroids, such as vector matroids or graphic matroids. We are particularly interested here in the problem of determining whether a given matroid can be expressed as a vector matroid over a finite field.
 
\begin{definition}
A matroid is representable over a field $F$ if it is isomorphic to a vector matroid over the field $F$.
A $\F_2$ representable matroid is called a binary matroid. 
\end{definition}

Note that there are several equivalent representation matrices of a given representable matroid over a field. 
It easy to show that on a rank $R$ matroid which is representable on $F$, one can always pick a representative matrix of the form $[I_R | A]$, where $A$ is an $R \times (n-R)$ matrix. This is called a standard representative matrix. 

We review here some basic construction defined on matroids. 
\begin{definition}
Let $M$ be a matroid with ground set $E$ and independent sets $\I$. Let $S \subseteq E$.
\begin{itemize}
\item  The restriction of $M$ to $S$, denoted by $M|S$, is the matroid whose ground set is $S$ and whose independent sets are the independent sets in $\I$ which are contained in (or equal to) $S$.
\item The contraction of $M$ by $S$, denoted by $M / S$, is the matroid whose ground set is $E-S$ and whose independent sets are the subsets $I$ of $E-S$ for which there exists a basis $B$ of $M|S$ such that $I \cup B \in \I$.
We will see an equivalent definition of the contraction operation when defining the dual of matroid. 
\item A matroid $N$ that is obtained from $M$ by a sequence of restrictions and contractions is called a minor of $M$. 
\end{itemize}
\end{definition}

We now define a matroid which is particular with respect to the representability theory of binary matroids.
\begin{definition}
Let $m,n \in \mZ_+$ with $m \leq n$. Let $E$ be a finite set with $n$ elements and $\B$ the collection of $m$-element subsets of $E$. One can easily check that this collection determines the bases of a matroid on $E$. We denote this matroid by $U_{m,n}$ and call it the uniform matroid of rank $m$ on a set with $n$ elements.
\end{definition}

The following are two major theorems concerning the representation theory of binary matroids. 
\begin{thm}\label{tutte}[Tutte]
A matroid is binary if and only if it has no minor that is $U_{2,4}$.
\end{thm}
\begin{thm}\label{whitney}[Whitney]
A matroid is binary if and only if the symmetric sum ($\triangle$) of any two circuits is the union of disjoint circuits.
\end{thm}

\begin{remark}\label{rem}
In a binary matroid, the circuit space of a matroid is equal to the kernel of its representative matrix. Indeed, if we multiply a circuit $C$ by the representative matrix $A$, we are summing the columns corresponding to a circuit. But this sum must be 0, since a circuit is a minimal dependent set, and therefore, we can express one of the columns as the sum of the others. 
\end{remark}

Next, we introduce the duality subject, 
which will play a central role in the applications of our next section.  

\begin{thm}
Let $M$ be a matroid on $E$ with a set of bases $\B$. Let $\B^* = \{E - B: B \in \B\}$. Then $\B^*$ is the set of bases of a matroid on $E$. 
We denote this matroid by $M^*$ and call it the dual of $M$.
\end{thm}

\begin{lemma}
If $r$ is the rank function of $M$, then the rank function of $M^*$ is given by
$$r^* (S) = r(S^c) + |S| - r(E).$$
\end{lemma}

We can then define the contraction operation via duality. 
\begin{definition}
The contraction of $M$ by $S$ is given by the dual of the restriction of $M^*$ on $S$, i.e., $M / S = (M^*|S)^*$.
\end{definition}

We conclude this section with the definition of polymatroids. 
\begin{definition}
A polymatroid is an ordered pair of a finite set $E$ called the ground set and a $\beta$-rank function $\rho: \Po(E) \rightarrow \mR_+$ which satisfies
\begin{align*}
 (R1) \quad &\text{$f(\emptyset) =0$} .\\
(R2) \quad &\text{If $X_1 \subseteq X_2 \subseteq E$, then $ f(X_1) \leq f(X_2)$}. \\
 (R3) \quad & \text{If $X_1 , X_2 \subseteq E$, then } \\
  & f(X_1 \cup X_2) + f(X_1 \cap X_2 ) \leq f(X_1) + f(X_2).
\end{align*}
 The region of $\mR^m$ defined by $\{(R_1,\dots,R_m): R_S \leq f(S), S \subseteq E\}$ is called a polyhedron. 
\end{definition}

We refer to \cite{oxley} for more details on matroid theory. 

\section{Extremal Dependencies}
This section connects the two previous ones, by characterizing MACs having an integer valued UMIF, i.e., a matroidal UMIF. Note that there exists a wide class of problems connecting information and matroid theory, such as characterizing entropic matroids; we refer to \cite{fuji, han, lovasz, matus, hanly, yeung} and references therein, and we show in Section \ref{entmat} that entropic matroids are particular cases of MAC matroids. 
An application of the results presented here is given in \cite{mmac}, for a MAC polar code construction.

Recall that $E_m=\{1,\ldots,m\}$. 

\begin{thm}[\cite{fuji}]
For any $m\geq 1$, any MAC $W \in M(\Y|\X^m)$ and any $P_1,\ldots, P_m \in M(\X)$ 
the function $\rho=I(P_1,\ldots,P_m,W)$ defined in \eqref{Ifct} is a $\beta$-rank function on $E_m$ and $(E_m,\rho)$ is a polymatroid.  

We denote this polymatroid by $M[P_1,\ldots,P_m,W]$. 
We use $M[W]$ when $P_1=\ldots=P_m=U_\X$.
If for a polymatroid $M$ we have $M\cong M[W]$ (where $\cong$ means isomorphic), we say that $W$ is a representative channel of $M$. 
\end{thm}



In this section, we are interested in characterizing the MACs for which the function $\rho$ is integer valued, i.e., for which $(E_m, \rho)$ defines a matroid. 
We restrict ourselves to binary MACs and we only consider the case where $P_1,\ldots,P_m$ are all given by the uniform distribution. 
One can easily come up with examples of binary MACs that would provide an integral $\rho$. 
But we are mostly interested in the reverse problem, i.e., in characterizing the matroids that admit such a mutual information representation. 
From a communication point of view, such MACs are interesting because they are trivial to communicate over with respect to both noise and interference management, and they indeed correspond to the extremal MACs created in the polarization process of \cite{mmac}.
 
\begin{definition}
A matroid $M$ is a BUMAC matroid if $M \cong M[W]$ for a binary MAC $W$.  
Hence, a BUMAC matroid is a matroid whose rank function is given by the UMIF (Definition \ref{mif}) of a binary MAC. ``BUMAC" refers to binary uniform MAC. 
\end{definition}

\begin{thm}\label{iso}
A matroid is BUMAC if and only if it is binary. 
\end{thm}
To prove this theorem, we first prove the following lemma.

\begin{lemma}\label{u24}
$U_{2,4}$ is not BUMAC.
\end{lemma}
\begin{proof}
Assume that the rank function of $U_{2,4}$ is the UNIF of a MAC. We then have
\begin{align}
& I(X[i,j]; Y) = 0 , \label{c1} \\ 
& I(X[i,j]; Y X[k,l]) = 2,  \label{c2}
\end{align}
for all $i,j,k,l$ distinct in $\{1,2,3,4\}$.
Let $y_0$ be in the support of $Y$. For $x\in \F_2^4$, define $\pp(x|y_0) = W(y_0|x)  / \sum_{z \in \F_2^4}W(y_0|z)$.
Then from \eqref{c2}, $\pp(0,0,*,*|y_0)=0$ for any choice of $*,*$ which is not $0,0$ and $\pp(0,1,*,*|y_0)=0$ for any choice of $*,*$ which is not $1,1$. On the other hand, from \eqref{c1},  $\pp(0,1,1,1|y_0)$ must be equal to $p_0$. However, we have form \eqref{c2} that $\pp(1,0,*,*|y_0)=0$ for any choice of $*,*$ (even for $1,1$ since we now have $\pp(0,1,1,1|y_0)>0$). At the same time, this implies that the average of $\pp(1,0,*,*|y_0)$ over $*,*$ is zero. This brings a contradiction, since from \eqref{c1}, this average must equal to $p_0$.
\end{proof}

\begin{proof}[Proof of Theorem \ref{iso}]
We start with the converse. Let $M$ be a binary matroid on $E$ with representative matrix $A$. Let $D$ be the deterministic channel defined by the matrix $A$, then we clearly have $M \cong M[D]$. 

For the direct part, let $M$ be a BUMAC matroid. 
We already know from Lemma \ref{u24} that $M$ cannot contain $U_{2,4}$ as a minor.
If instead $U_{2,4}$ is obtained by a contraction of $S^c$ from $M$ , i.e., $M /S^c \cong U_{2,4}$, it means that $(M^* | S)^* \cong U_{2,4}$. Since $U_{2,4}$ is self dual, we have $M^* | S \cong U_{2,4}$. Let us denote by $r^*$ the rank function of $M^*$. 
We have for any $Q \subseteq E$
\begin{align*}
r^* (Q)&= |Q|+r(Q^c) - r(E), \\
&= |Q|+I(X[Q^c]; Y X[Q]) - I(X[E];Y), \\
& = |Q| - I(X[Q];Y) ,
\end{align*}
where the last equality follows from the chain rule of the mutual information. 
Since $r^*(\cdot)$ restricted to $S$ is the rank function of $U_{2,4}$, we have in particular 
\begin{align*}
& r^*(T) = 2 , \quad \forall T \subset S \text{ s.t. } |T|=2\\
& r^*(S)=2,
\end{align*}
that is, 
\begin{align}
&2- I(X[T];Y) = 2, \label{d1} \quad \forall T \subset S \text{ s.t. } |T|=2 \\
& 4-I(X[S]; Y)=2. \notag
\end{align}
 This implies, by the chain rule of the mutual information, 
\begin{align}
& I(X[T]; Y X[S-T]) = 2 \label{d2}, \quad \forall T \subset S \text{ s.t. } |T|=2.
\end{align}



Hence, from the proof of Lemma \ref{u24}, \eqref{d1},\eqref{d2} cannot simultaneously hold and $U_{2,4}$ cannot be a minor of $M$. From Tutte's Theorem (cf. Theorem \ref{tutte}), $M$ is binary.
\end{proof}

Previous theorem gives a characterization of BUMAC matroids. Note that, if we were interested in characterizing binary matroids through BUMAC matroids, then the following corollary holds.

\begin{definition}
A BULMAC matroid is a BUMAC matroid with linear deterministic representative channel. 
\end{definition}

\begin{corol}
The family of binary matroids is isomorphic to the family of BULMAC matroids.
\end{corol}

We now formally establish the connection between extremal MACs and linear deterministic MACs.

\begin{thm}\label{equiv}
Let $W$ be a binary MAC with $m$ users whose UMIF is integer valued, i.e., $M[W]$ is a binary matroid.
Let $A$ be a matrix representation of $M[W]$ and let $Y$ be the output of $W$ when the input $X[E_m]$ (with i.i.d. uniform components) is sent. 
Then $$I(A X[E_m];Y)= \mathrm{rank} A=I(X[E_m];Y).$$
\end{thm}
This theorem says that for a binary MAC with integer valued UMIF, the output of i.i.d.\ uniform inputs contains all the information about the corresponding linear form of the inputs and nothing more. In that sense, MACs with integer valued UMIF are ``equivalent'' to linear deterministic MACs. 
\begin{proof}
Let $M=M[W]$ with $M[W] \cong M[A]$ and let us assume that $M$ has rank $R$. Let $\mathcal{B}$ be the set of bases of $M$ and let $\mathcal{B}^*$ be the set of bases of $M^*$. 
Since $r(B)=|B|=R$ for any $B \in \mathcal{B}$, we have
\begin{align}
r(B)= I(X[B];YX[B^c]) = R, \quad \forall B \in \mathcal{B}. \label{base}
\end{align}
Moreover, the rank function of $M^*$ is given by $$r^*(S)=|S|-I(X[S];Y)$$ and for all $D \in \mathcal{B}^*$, we have $r^*(D) = |D|=|E_m|-R$. Hence
\begin{align*}
r^*(D) &= |E_m|-R =  |E_m|-R - I(X[D];Y),\quad \forall D \in \mathcal{B}^* , 
\end{align*}
or equivalently
\begin{align}
 I(X[D];Y) = 0,\quad \forall D \in \mathcal{B}^* .  \label{base*}
\end{align}
Hence, form \eqref{base} and \eqref{base*} and the fact that $\mathcal{B}^* = \{E_m - B : B \in \mathcal{B}\}$, we have
\begin{align}
& I(X[B];YX[B^c]) = r, \quad \forall B \in \mathcal{B}, \label{ba1} \\
& I(X[B^c];Y) = 0,\quad \forall B \in \mathcal{B}. \label{ba2}
\end{align}

Note that \eqref{ba1} means that if any realization of the output $Y$ is given together with any realization of $X[B^c]$, we can determine $X[B]$. Moreover, \eqref{ba2} means that $X[B^c]$ is independent of $Y$. 
Let us analyze how these conditions translate in terms of probability distributions. 
Let $y_0 \in \mathrm{Supp}(Y)$.
We define
 $$p_0(x) := W(y_0|x)  / \sum_{x^{'} \in \F_2^m}W(y_0|x^{'}), \quad \forall x\in \F_2^m.$$
From \eqref{ba1}, if $p_0 (x)>0$, we must have 
$p_0(x^{'})=0$ for any $x^{'}$ such that $x^{'}[ B^c]=x[ B^c] $ for some $B^c \in \mathcal{B}^*$. From \eqref{ba2}, we have that 
$$\sum_{x^{'}: x^{'}[B^c]=x[B^c]} p_0(x) = 2^{R-m}, \quad \forall B \in \mathcal{B}, x [B^c] \in \F_2^{m-R}.$$
Hence, for any $B \in \mathcal{B}$ and any $x[B^c] \in \F_2^{m-R}$, we have
\begin{align}
& \bigvee_{x^{'} : x^{'}[B^c] = x[B^c]} p_0(x^{'}) = 2^{R-m}, \label{e1} \\
& \sum_{x^{'} : x^{'}[B^c] = x[B^c]} p_0(x^{'}) = 2^{R-m} .\label{e2}
\end{align}
Let $\star:=2^{R-m}$.
Previous constraints imply that $p_0(x) \in \{0, \star\}$ for any $x\in \F_2^{m}$ and that the number of $x$ with $p_0(x)= \star$ is exactly $2^{|E_m|-r}$. 
Let us assume w.l.o.g. that $p_0(\bar{0}) =\star$, where $\bar{0} $ is the all 0 vector. 
Note that we know one solution that satisfies previous conditions. Namely, the solution that assigns a $\star$ to all vectors belonging to $\mathrm{Ker} A$. As expected, $  \mathrm{dim} \mathrm{Ker} (A) = |E_m| - \mathrm{rank} (A) = |E_m| - r$. 
We want to show that there cannot be any other assignment of the $\star$'s in agreement with the matroid $M$. 
In the following, we consider elements of $\F_2^m$ as binary vectors or as subsets of $E_m$, since $\F_2^m\cong 2^{E_m}$. The field operations on $\F_2^m$ translate into set operations on $2^{E_m}$, in particular, the component wise modulo 2 addition $x_1 + x_2$ of binary vectors corresponds to the symmetric different $x_1 \triangle x_2$ of sets, and the component wise multiplication $x_1 \cdot x_2$ of binary vectors corresponds to the intersection $x_1 \cap x_2$ of sets.


We now check which are the assignments which would not violate \eqref{e1} and \eqref{e2}.
We have assumed w.l.o.g that $\bar{0}$ is assigned $\star$, hence $\emptyset$ is assigned $\star$. From \eqref{e1}, any $x$ for which $x[B^c]=0$ for some $B \in \B$,
 must be assigned 0. Note that
 \begin{align*}
  x[B^c]=0 \equiv  x \cdot B^c =0 \equiv  x \subseteq B \equiv  x \in \I, 
 \end{align*}
where $\I$ is the collection of independent sets of $M$.
Hence, the elements which are assigned 0 by checking the condition \eqref{e1} are the independent sets of $M$, besides $\emptyset$ which is assigned $\star$.

For $B \in \B$ and $s \in \F_2^m$, we define
$$\I(B) := \{ I \in \I  : I \subseteq B\}$$
and
$$\I_{s}(B) := \{ x:  \, x[B^c] = s[B^c]\}.$$
Note that $\I_{s}(B) = s+ \I (B) ,$ indeed:
 \begin{align*}
 x[B^c]=s[B^c]  \equiv  x \cdot B^c = s \cdot B^c \equiv  (x + s) \cdot B^c =0 \equiv  x + s \subseteq B 
\equiv  x \in s+\I (B). 
 \end{align*}



Now, if $r(S) = r(T)$ for two sets $S$ and $T$ with $T\subseteq S$, we have
$$I(X[S-T]; Y X[S^c]) = 0.$$
This means that $(Y,X[S^c])$ is independent of $X[S-T]$. From the point of view of probability distributions, this means that
\begin{align*}
\pp_{X[S-T] | Y X[S^c]} (x[S-T] | y_0 x [S^c]) = \frac{1}{2^{|T|}}, \quad \forall x[S-T] , x [S^c]
\end{align*}
or equivalently, 
\begin{align*}
\sum_{S - T} p_0(x[E]) = \frac{1}{2^{|T|}} \sum_{S } p_0(x[E]), \quad \forall x[T] , x [S^c] .
\end{align*}
Hence, if we set the components of $x\in \F_2^m$ to frozen values on $S^c$, then, no matter how we freeze the components of $x$ on $T$, the average of $p_0(\cdot)$ on $S-T$ must be the same. 

Let $C \in \C$ be a circuit. By the definition of circuits, if we remove any element of $C$ we have a basis of $M|C$. Let $B$ be such a basis, we then have $r(C) = r(B)$. We now want to freeze the values on $C^c$ and $B$ in two ways. 

1. If we pick $d=C \cap B^c$, then
$$I_d(B) = \{ x: x \subseteq C- B\}.$$
These are the elements that are strictly contained in $C$, i.e., elements of $\I$, including $\emptyset$. Therefore, the average of $p_0(\cdot)$ must be $\star$ for this freezing. 

2. If we pick $d=C$, we already know that the average of $p_0(\cdot)$ must be $\star$, but we have
$$I_d(B) = \{ x: x+ C \subseteq C- B\}.$$
These are the elements containing $B$, possibly elements of $C-B$ but nothing else. Therefore, the options are $x= C$ or $x \in \I -  \emptyset$. This forces $C$ to be assigned $\star$.

Hence, we have shown that all circuits of $M$ must be assigned $\star$. 
This in turns imply several other 0 assignments. Namely,
\begin{align}
\bigcup_{C \in C} C + (\I - \emptyset) \label{ck}
\end{align}
must be assigned $0$.

Let us next consider a union of two disjoint circuits, $D = C_1 \sqcup C_2$. 
Then, if we remove any single elements of $D$, say by removing an element of $C_1$, we obtain a disjoint union of an independent set and a circuit, say $I \sqcup C_2$. 
Hence, 
$$r(C_1 \sqcup C_2) = r(I \sqcup C_2).$$
We can then use the same technique as previously, but this time, we need to use that \eqref{ck} is assigned 0. Note that is important to assume that the union is disjoint, in order to guarantee that $C_2 + I \sqcup C_2 = I \in \I $.

We can then use an induction to show that any union of disjoint circuits must be assigned $\star$. 
Finally, for a binary matroid, any symmetric difference of two circuits is given by a union of disjoint circuits (this can be directly checked but notice that it is contained as one of the implications of Theorem \ref{whitney} due to Whitney). Hence, the space generated by the circuits, seen as elements of $(\F_2^m, +)$ must be assigned $\star$,
and using Remark \ref{rem}, we conclude the proof since we have assigned the $1/\star$ numbers of $\star$ without any degrees of freedom, and the assignment is done on $\mathrm{Ker} A$. 
\end{proof}

\subsection{Recursions using mutual information properties}
In this section we re-derive some of the results of previous section using inductive arguments. 
We start by checking a result similar to Theorem \ref{equiv} for the case $m=3$.

\begin{lemma}\label{3}
Let $W$ be a binary MAC with 2 users. Let $X[E_2]$ with i.i.d. uniform binary components and let $Y$ be the output of $W$ when $X[E]$ is sent.
If $I(X[1];YX[2])$, $I(X[2];YX[1])$ and $I(X[1]X[2];Y)$ have specified integer values, then $I(X[1];Y), I(X[2];Y)$ and $I(X[1]+X[2];Y)$ have specified values in $\{0,1\}$, and vice-versa.
\end{lemma}
\begin{proof}
Let 
\begin{align*}
&I:=[I(X[1];YX[2]), I(X[2];YX[1]), I(X[1]X[2];Y)]\\
&J:=[I(X[1];Y), I(X[2];Y), I(X[1]+X[2];Y)].
\end{align*}
Note that by the polymatroid property of the mutual information, we have
\begin{align}
I \in \{[0,0,0],  [0,1,1], [1,0,1],  [1,1,1],  [1,1,2]\}. \label{poss}
\end{align}
Let $y \in \supp (Y)$ and for any $x\in \F_2^2$ define $\pp(x|y) = W(y|x)  / \sum_{z \in \F_2^2}W(y|z)$ (recall that $W$ is the MAC with inputs $X[1],X[2]$ and output $Y$). 
Assume w.l.o.g. that $p_0:=\pp(0,0|y)>0$. 

\begin{itemize}
\item If $I=[0,0,0]$ we clearly must have $J =[0,0,0]$.

\item If $I=[*,1,1]$, we have $I(X[2];YX[1])=1$ and we can determine $X[2]$ by observing $X[1]$ and $Y$, which implies 
$$\pp (01|y)=0.$$
Moreover, since $I(X[1];Y)= I(X[1]X[2];Y) - I(X[2];Y X[1])=0$, i.e., $X[1]$ is independent of $Y$, we must have that $\sum_{x[2]}\pp(x[1]x[2]|y)$ is uniform, and hence, 
\begin{align*}
& \pp (00| y ) = 1/2,
& \pp (10| y ) + \pp (11| y )=1/2. 
\end{align*}
Now, if $\star = 1$, by a symmetric argument as before, we must have $\pp (11| y )=1/2$ and hence we the input pairs $00$ and $11$ have each probability half (a similar situation occurs when assuming that $\pp(x|y)>0$ for $x\neq (0,0)$), and we can only recover $X[1]+X[2]$ from $Y$, i.e., $J=[0,0,1]$.
If instead $*=0$, we then have $I(X[2];Y)= I(X[1]X[2];Y) - I(X[1];Y X[2])=1$ and from a realization of $Y$ we can determine $X[2]$, i.e., $\pp(10)=1/2$ and $J=[0,1,0]$.

\item If $I=[1,0,1]$, by symmetry with the previous case, we have $J=[1,0,0]$.

\item If $I=[1,1,2]$, we can recover all inputs from $Y$, hence $J =[1,1,1]$.

\end{itemize}

For the converse statement, 
Note that $J$ must be given by $[0,0,0], [0,1,0], [1,0,0], [0,0,1]$ or $[1,1,1]$. Clearly, the case $[0,0,0]$ implies $I=[0,0,0]$. 

For the case $J=[0,1,0]$, note that $I(X[2];Y)=1$ implies $h(X[2]|Y)=0$, i.e., for any $y \in \supp(Y)$, $h(X[2]|Y=y)=0$.
This means that for any $y \in \supp(Y)$, if $p_2(x[2]|y)>0$ for some $x[2]$, we must have $p_2(\tilde{x}[2]|y)=0$ for $\tilde{x}[2] \neq x[2]$. We use $p_i$, $i=1,2$, for the probability distribution of $X[i]$ given the realization $Y=y$ and $p_{12}$ for the probability distribution of $(X[1],X[2])$ given $Y=y$. Assume now (w.l.o.g.) that $p_{12}(0,0|y)>0$. Since $p_2(x[2]|y)=\sum_{x[1]}p_{12}(x[1] x[2]|y)$, previous observation implies that $p_{12}(0 1|y)=p_{12}(1 1|y)=0$. Moreover, $I(X[1];Y)=0$ implies that $h(X[1]|Y=y)=1$, i.e., for any realization of $Y$, the marginal of $X[1]$ is uniform, which implies $p_{12}(0 0|y)=p_{12}(1 0|y)=1/2$. Hence, if we are given the realization of $X[1]$ and $Y$, we can decide what $X[2]$ must be, and this holds no matter which values of $(X[1],X[2])$ is assigned a positive probability, i.e., $I(X[2];YX[1])=1$. If instead we are given $X[2]$ and $Y$, we can not infer anything about $X[1]$, i.e., $I(X[1];YX[2])=0$. Finally, by the chain rule, $I(X[1]X[2];Y)=1$. The case where $[I(X[1];Y), I(X[2];Y),I(X[1]+X[2];Y)]$ is equal to $[1,0,0]$ can be treated symmetrically and the other cases in a similar fashion.
\end{proof}

\begin{lemma}\label{other}
Let $W$ be a binary MAC with $m$ users. Let $X[E_m]$ with i.i.d. uniform binary components and let $Y$ be the output of $W$ when $X[E]$ is sent.
If $I(X[S];YX[S^c])$ has a specified integer value for any $S \subseteq E_m$, 
then $I(X[E_m]\cdot S;Y)$ has a specified value in $\{0,1\}$ for any $S \subseteq E_m$, and vice-versa.  
Note: $X[E_m]\cdot S= \oplus_{i\in S} X[i]$
\end{lemma}
The recursive argument for the proof of the direct part of this Lemma has been proposed by Eren \c{S}a\c{s}o\u{g}lu  \cite{eren} and contains the idea behind this section. The direct statement in the Lemma is a consequence of Theorem \ref{equiv} but is proved here using the recursive approach.

\begin{proof}
Let $I[S](W)$ be assigned an integer for any $S \subseteq \M$.
By the chain rule of the mutual information
$$I(X[\M]; Y) = I(X[S]; Y) + I(X[S^c] ; Y X[S]),$$
and we can determine $I(X[S]; Y)$ for any $S$. Since for any $T \subseteq S$
$$I(X[S]; Y) = I(X[T]; Y) + I(X[S-T] ; Y X[T]),$$ 
we can also determine $I(X[S]; Y X[T])$ for any $S,T \subseteq \M$ with $S \cap T = \emptyset$.
Hence we can determine 
\begin{align*}
&I(X[1] ,X[2]; Y X[S]) \\
& I(X[1]; Y X[S] X[2]) \\
& I(X[2] ; Y X[S] X[1] ) 
\end{align*}
and using Lemma \ref{3}, we can determine
\begin{align*}
&I(X[1] + X[2] ; Y X[S]) 
\end{align*}
for any $S \subseteq \M$ with $\{1,2\} \notin S$, hence 
\begin{align*}
&I(X[i] + X[j] ; Y) 
\end{align*}
for any $i,j \in \M$.

Assume now that we have determined $I(\sum_T X[i] ; Y X[S])$ for any $T$ with $|T| \leq k$ and $S \subseteq \M -T$. 
Let $T=\{1,\dots, k\}$ and let $S \subseteq \{k+2,\dots,m\}$.
\begin{align*}
&I(\sum_T X[i], X[k+1]; Y X[S] ) \\&= I( X[k+1]; Y X[S]) + I(\sum_T X[i] ; YX[S] X[k+1]),
\end{align*}
in particular, we can determine
\begin{align*}
& I(X[k+1] ; Y \sum_T X[i], X[S] ) \\
&= I(\sum_T X[i] , X[k+1]; Y X[S] )\\& - I(\sum_T X[i]; YX[S] )
\end{align*}
and
\begin{align*}
&I(\sum_T X[i] ,X[k+1]; Y X[S]) \\
& I(\sum_T X[i] ; Y X[S] X[k+1]) \\
& I(X[k+1] ; Y \sum_T X[i] ,X[S] ) 
\end{align*}
and using Lemma \ref{3}, we can determine
\begin{align*}
&I(\sum_T X[i] + X[k+1]; Y X[S])
\end{align*}
hence
\begin{align*}
&I(\sum_{T}X[i]; Y)
\end{align*}
for any $T \subseteq \M$ with $|T|= k+1$.
Hence, inducting this argument, we can determine $I(\sum_{T}X[i]; Y)$ for any $T \subseteq \M$.

For the converse statement, assume that we are given $I(X[E_m]\cdot S;Y) \in \{0,1\}$ for any $S \subseteq E_m$.
In particular, $I(X_i;Y)$, $I(X_i;Y)$ and $I(X_i +X_j;Y)$ is determined for any $i,j \in E_m$, and hence, from Lemma \ref{3}, we have that $I(X_i;YX_i)$ and $I(X_i X_j;Y)$ are determined (and integer valued) for any $i,j \in E_m$.

Note that we can also determine $I(X[E_m] \cdot T ; Y X[i])$ for any $T \subset E_m$ and $i \in E_m-T$; indeed, 
we know $I(X[E_m] \cdot T; Y)$ for any $T\subset E_m$, so for $i \in E_m- T$, we know
\begin{align}
&I(X[i] + X[E_m] \cdot T ; Y),\\
&I(X[i]  ; Y),\\
& I(X[E_m] \cdot T ; Y),
\end{align}
and hence, using Lemma \ref{3}, we can determine $I(X[E_m] \cdot T ; Y X[i])$.

Let us assume now that we have determined $I(X[S];Y X[F-S])$ for any $F$ such that $|F| \leq k$ and $S \subseteq F$, as well
as $I(X[E_m] \cdot T; Y X[K])$ for any $K$ such that $|K| \leq k-1$ and $T \subseteq E_m-K$. We have already checked that this can be determined for $k=2$. We now check that we can also determine these quantities for $k+1$ instead of $k$.

Let $K$ with $|K|=k-1$. Assume w.l.o.g. that $1,2,3 \notin K$. Since we assume to know 
\begin{align}
&I(X[1]; Y X[K]),\\
&I(X[1] + X[2]; Y X[K]), \\
&I(X[1] + X[2] + X[3]; Y X[K]),
\end{align}
using Lemma \ref{3}, we can determine $I(X[1]+X[2]; Y X[K \cup 3])$. Using a similar argument we can determine $I(X[E_m] \cdot T; Y X[K])$ for any $K$ such that $|K| \leq k$ and $T \subseteq E_m-K$.
Moreover, since we now know $I(X[1]+X[2]; Y X[K ])$ and also 
\begin{align}
&I(X[1]; Y X[K]),\\
&I( X[2]; Y X[K ]), 
\end{align}
we can determine with Lemma \ref{3}
\begin{align}
&I(X[1]; Y X[K \cup 2]),\\
&I(X[2]; Y X[K \cup 1 ]),\\
&I(X[1]X[2]; Y X[K ]),
\end{align}
and hence, we can determine $I(X[K_1]; Y X[K_2])$ for $|K_1| \leq 2$ and $|K_1|+|K_2| \leq k+1$.
From the chain rule of the mutual information, we have
\begin{align}
&I(X[1]X[2] X[3]; Y X[K-3]) = I(X[1] X[2] ; Y X[K-3]) +I(X[3] ; Y X[K -3] X[1]X[2]) 
\end{align}
and both term in the right hand side above are already determined.
Hence, by iterating the chain rule argument, we can determine
$I(X[S];Y X[F-S])$ for any $F$ such that $|F| \leq k+1$ and $S \subseteq F$.
Finally, we can iterate these arguments on $k$ to reach $F=E_m$, i.e., to determine an integer for $I(X[S];YX[S^c])$ for any $S \subseteq E_m$.

\end{proof}

\subsection{Quasi-Extremal Channels}
In this section, we provide technical steps necessary to extend previous results to polymatroids which are ``close'' to matroids. 

\begin{lemma}
Let $W$ be a binary MAC with $m$ users. Let $X[E_m]$ with i.i.d. uniform binary components and let $Y$ be the output of $W$ when $X[E]$ is sent. Let $\e>0$, if 
$I(X[S];YX[S^c])$ has a specified value in $\mZ + (-\e, \e)$ for any $S \subseteq E_m$, then
$I(X[E_m]\cdot S;Y)$ has a specified value in $[0,o_\e(1)) \cup (1-o_\e(1),1]$ for any $S \subseteq E_m$.
Note: $X[E_m]\cdot S= \oplus_{i\in S} X[i]$
\end{lemma}
The converse of this statement also holds.
This lemma follows from the results of previous sections and from the following lemmas.

\begin{lemma}
For two random variables $X,Y$ such that $X$ is binary uniform and $I(X;Y) < \e$, we have
$$ \Pr \{ y : \| P_{X|Y}(\cdot|y)- U(\cdot) \|_1 < \e^{1/2} \} \geq 1- 2 \ln 2  \, \e^{1/2},$$
where $U$ is the binary uniform measure. 
\end{lemma}

\begin{proof}
Since $I(X;Y) < \e$, we have
$$D(P_{XY} || P_X P_Y) < \e$$
and from Pinsker's inequality
$$ \frac{1}{2 \ln 2} \| P-Q\|_1 \leq D(P||Q)$$
we get
$$ \| P_{XY} - P_X P_Y \|_1 = \sum_{y} P_Y (y) \| P_{X|Y}(\cdot|y) - U(\cdot) \|_1 \leq 2 \ln 2 \, \e.$$
Therefore, by Markov's inequality, we have
\begin{align*}
\Pr \{ y: \|  P_{X|Y}(\cdot|y) - U(\cdot) \|_1 \geq  a \} \leq \frac{ 2 \ln 2 \,\e}{a}
\end{align*}
and by choosing $a=\e^{1/2}$, we get the desired inequality. 
\end{proof}

\begin{lemma}
For two random variables $X,Y$ such that $X$ is binary uniform and $h(X|Y) < \e$, define $E_\e$ by
$$y \in E_\e \quad \Longleftrightarrow \quad \Pr \{ X=0 | Y=y \} \Pr \{ X=1 | Y=y \} \leq \e,$$
then $$\Pr \{ E_\e \} \geq 1- \gamma (\e),$$
with $\gamma(\e) \to 0$ when $\e \to 0$.
\end{lemma}
This lemma tells us that if $\Pr \{ X=0 | Y=y \}$ is not small, we must have that $\Pr \{ X=1 | Y=y \}$ is small with high probability. It is given as a problem in \cite{gallager}.

\section{Extensions}

\subsection{$q$-ary Matroids}
The results of last sections are expected to generalize to the $q$-ary alphabet case, where $q$ is a prime or power of prime.
In particular, we have the following. 

{\bf Claim:}
A matroid is $q$-ary representable if and only if its rank function is given by the UMIF of a MAC with $q$-ary inputs.

Hence, one could equivalently characterize $q$-ary matroids by characterizing rank functions which are representable by $q$-ary alphabets MAC. 

\subsection{Entropic matroids}\label{entmat}
The following result can be found in \cite{fuji}.
\begin{lemma}
Let $Z[E_m]$ be an $m$-dimensional random vector with arbitrary distribution over $\F_q^m$.
Then the function $r:S \mapsto H(Z[S])$ is a $\beta$-rank function and $(E_m,r)$ is a polymatroid.
\end{lemma}

Hence, we can define a notion of {\it entropic matroid}, which is a matroid whose rank function is representable by an entropic function as above. 

We now show that entropic matroids can be studied as specific cases of MAC matroids.
Consider a specific MAC which consist of an additive noise perturbation of the input, i.e.,
$$Y[E_m] = X[E_m] \oplus Z[E_m],$$
where $X[E_m]$ is an $m$-dimensional random vector with i.i.d.\ uniform components over $\F_q$ and 
$Z[E_m]$ is an $m$-dimensional random vector of arbitrary distribution on $\F_q^m$, independent of $X[E_m]$. 
Then,
\begin{align*}
I(X[S]; Y X[S^c] ) & = |S| - H(X[S] | Y X[S^c]) \\
& = |S| - H(Y[S] \ominus Z[S] | Y, (Y[S^c]\ominus Z[S^c])) \\
& = |S| - H(Z[S] | Y, Z[S^c]) \\
& = |S| - H(Z[S] |Z[S^c])\\
& = H(Z[S]) + |S|- H(Z[E_m]) .
\end{align*} 
Hence, an entropic matroid corresponds to a particular case of MAC matroid which has additive noise. 









\end{document}